\documentclass[11pt]{amsart}
\usepackage{amssymb, amsfonts, amsthm, amscd, mathrsfs, hyperref}
\usepackage{graphicx, xcolor}
\usepackage[all]{xy}
\usepackage[leqno]{amsmath}

\theoremstyle{definition}
\newtheorem{thm}{Theorem}[section]

\newtheorem{cor}[thm]{Corollary}
\newtheorem{prop}[thm]{Proposition}
\newtheorem{dfn}[thm]{Definition}

\newtheorem*{rmk}{Remark}
\theoremstyle{definition}

\makeatletter
\def\endo{\mathrm{End}}

\def\id{\mathrm{id}}

\def\tr{\mathrm{Tr}}
\def\ccc{\mathbb{C}}

\def\rr{\mathbb{R}}

\def\ad{\mathrm{ad}}

\def\frg{\mathfrak{g}}

\def\frM{\mathfrak{M}}
\def\pt{\partial}
\def\bpt{\overline{\pt}}

\def\ud{\mathrm{d}}
\def\spann{\mathrm{span}}

\def\be{\overline{e}}

\def\spl{\mathfrak{sl}}

\def\re{\mathrm{Re}}
\def\ima{\mathrm{Im}}

\makeatother

\title{\textbf{Invariant Solutions to the Strominger System on Complex Lie Groups and Their Quotients}}
\date{}
\author{Teng Fei and Shing-Tung Yau}

\begin{document}
\maketitle{}
\tableofcontents

\section{Introduction}

In \cite{strominger1986}, Strominger analyzed heterotic superstring background with nonzero torsion by allowing a scalar ``warp factor'' for the spacetime metric. Consideration of supersymmetry and anomaly cancellation imposes a complicated system of PDEs on the internal manifold known as the Strominger system. Ever since then, there has been much effort devoted in finding solutions to the Strominger system. For threefolds, Strominger described some perturbative solutions in \cite{strominger1986}. Many years later, Li and Yau \cite{li2005} obtained the first smooth irreducible solutions to the system for $U(4)$ or $U(5)$ principal bundles on K\"ahler Calabi-Yau manifolds, which was further developed in \cite{andreas2012}. As for non-K\"ahler Calabi-Yau inner spaces, the first solution was constructed by Fu and Yau \cite{fu2008}. Later more non-K\"ahler solutions were found, especially on nilmanifolds, see \cite{fernandez2009}, \cite{grantcharov2011} and the references therein. Some local models were studied in \cite{fu2009}.

From a mathematical point of view, the Strominger system can be formulated as follows. Let $(X^n,g,J)$ be an Hermitian $n$-fold (not necessarily K\"ahler) with holomorphically trivial canonical bundle and let $\Omega$ be a nowhere-vanishing holomorphic $(n,0)$-form on $X$. We denote the positive $(1,1)$-form associated with $g$ by $\omega$ and the curvature form of $(T_\ccc X,g)$ with respect to certain Hermitian connection by $R$. In addition, let $(E,h)$ be a holomorphic vector bundle over $X$ and $F$ its curvature form with respect to the Chern connection. The Strominger system \footnote{Here we follow the formulation in \cite{li2005}.} consists of the following equations (mostly people are solely interested in the $n=3$ case):
\begin{eqnarray}
\label{hym}F\wedge\omega^{n-1}=0,\quad F^{0,2}=F^{2,0}=0,\\
\label{ac}\sqrt{-1}\pt\bpt\omega=\frac{\alpha'}{4}(\tr~R\wedge R-\tr~F\wedge F),\\
\label{cb}\ud(\|\Omega\|_\omega\cdot\omega^{n-1})=0.
\end{eqnarray}
From now on, we will call Equations (\ref{hym}), (\ref{ac}) and (\ref{cb}) the Hermitian-Yang-Mills equation, the anomaly cancellation equation and the conformally balanced equation respectively.

If $\omega$ is a K\"ahler metric, then Equation (\ref{cb}) implies that $\|\Omega\|_\omega$ is a constant. That is to say, $(X,g)$ has $SU(n)$-holonomy. From Yau's theorem \cite{yau1978}, we know that there is a unique such metric in the given cohomology class assuming that $X$ is compact.

For a general Hermitian manifold, Equation (\ref{cb}) implies that the rescaled metric $\tilde{\omega}=\|\Omega\|_\omega^{\frac{1}{n-1}}\cdot\omega$ is balanced, i.e., $\ud(\tilde{\omega}^{n-1})=0$, in the sense of Michelsohn \cite{michelsohn1982}. This condition imposes certain mild topological restriction for the internal manifold $X$, (see \cite{michelsohn1982} for the intrinsic characterization of balanced manifolds), which excludes, for instance, certain $T^2$-fiber bundles over Kodaira surface using a construction of Goldstein and Prokushkin \cite{goldstein2004}. As $\tilde{\omega}$ is balanced, it is also Gauduchon, i.e., $\pt\bpt(\tilde{\omega}^{n-1})=0$. Hence by the theorem of Uhlenbeck-Yau \cite{uhlenbeck1986} and Li-Yau \cite{li1987}, Equation (\ref{hym}) is equivalent to the statement that $E$ is poly-stable. Consequently, the main difficulty in solving Strominger system is to deal with the anomaly cancellation equation.

As an analogue of the K\"ahler situation we discussed, we can think of the Strominger system as a guidance on finding canonical metrics on balanced manifolds, at least for non-K\"ahler Calabi-Yau's, which further sheds light on understanding Reid's fantasy \cite{reid1987}. The Reid's proposal basically says all Calabi-Yau's are connected via conifold transition by going into the non-K\"ahler territory.

The prototype of conifold transition is the transformation between smoothing and deformation of the conifold $\{z_1^2+\dots+z_4^2=0\}\subset\ccc^4$. Therefore it is of vital importance to understand the Strominger system on the smoothing of the conifold, which can be identified with the complex semisimple Lie group $SL_2\ccc$.

In 2013, Biswas and Mukherjee published a paper \cite{biswas2013}, claiming that they have found an invariant solution to the Strominger system on $SL_2\ccc$. However, it was soon pointed out by Andreas and Garcia-Fernandez \cite{andreas2014} that there was an error in Biswas and Mukherjee's calculation and there is actually no solution to the Strominger system in that setting. Furthermore, Andreas and Garcia-Fernandez proposed looking for solutions to the Strominger system using Strominger-Bismut connection. Inspired by their idea, we are able to obtain a few interesting invariant solutions to the Strominger system on complex Lie groups and their quotients, which is the main result of this paper .

This paper is organized as follows. In Section 2 we briefly review the theory of Hermitian connections on an almost Hermitian manifold. Section 3 focuses in the flat (i.e., $F\equiv0$) case. Using the canonical 1-parameter family of Hermitian connections described in Section 2, we obtain a class of invariant solutions to the Strominger system on various complex Lie groups, giving a rather complete answer to the problem discussed by \cite{biswas2013} and \cite{andreas2014}. In Section 4 we take non-flat bundles $E$ into our consideration. In particular, for the $SL_2\ccc$ case, we construct invariant solutions to the Strominger system for trivial but non-flat bundle $E$ of any rank.

\section{The Canonical 1-parameter Family of Hermitian Connections}

As argued in \cite{strominger1986}, Strominger system requires the connection on $T_\ccc X$ to be Hermitian, i.e., it preserves both the metric $g$ and the complex structure $J$. A natural choice of such connection is the Chern connection. However, as shown in \cite{andreas2014}, the ansatz used by \cite{biswas2013} always yields $R=0$, violating the anomaly cancellation equation. Therefore we need to consider more general Hermitian connections other than Chern. Following \cite{gauduchon1997}, we will review the general theory of Hermitian connections and the construction of the canonical 1-parameter family of Hermitian connections.

Let $(X^n,g,J)$ be an almost Hermitian $n$-fold. Using the Riemannian metric $g$, we may identify any real $T_\rr X$-valued 2-form $B\in\Omega^2(T_\rr X)$ with a real trilinear form which is skew-symmetric with respect to the last two variables: \[B(U,V,W)=\langle U,B(V,W)\rangle\textrm{ for any vector fields } U,V,W.\] Let $\omega$ be the associated Hermitian form, we also introduce the real 3-form $\ud^c\omega$ by \[\ud^c\omega(U,V,W)=-(\ud\omega)(JU,JV,JW).\] If $J$ is integrable, $\ud^c$ coincides with the usual notation $\ud^c=\sqrt{-1}(\bpt-\pt)$. Let $\frM$ be the involution on $\Omega^2(T_\rr X)$ defined by \[(\frM B)(U,V,W)=B(U,JV,JW)\] and we denote the (+1)-eigenspace of $\frM$ by $\Omega^{1,1}(T_\rr X)$.

\begin{dfn}
A Hermitian connection $\nabla$ on $T_\rr X$ is an affine connection that preserves both the metric $g$ and the complex structure $J$, i.e., $\nabla g=0$ and $\nabla J=0$.
\end{dfn}
It is easy to see that the space of Hermitian connections forms an affine space modelled on $\Omega^{1,1}(T_\rr X)$.

The canonical 1-parameter family of Hermitian connections $\nabla^t$ is defined by \[\langle\nabla^t_UV,W\rangle=\langle D_UV,W\rangle+\frac{1}{2}\langle(D_XJ)JY,Z\rangle+\frac{t}{4}((\ud^c\omega)^+(U,V,W)+(\ud^c\omega)^+(U,JV,JW)),\] where $D$ is the Levi-Civita connection and $\alpha^+$ denotes the $(2,1)+(1,2)$-part of a 3-form $\alpha$.

\begin{thm}(\cite{gauduchon1997}) The canonical 1-parameter family of Hermitian connections forms an affine line. To be precise, it satisfies \[\nabla^t=\nabla^0+\frac{t}{4}((\ud^c\omega)^++\frM(\ud^c\omega)^+),\] where we have to identify the 3-form $(\ud^c\omega)^+$ as an element of $\Omega^2(TM)$. This affine line parameterizes all the known ``canonical'' Hermitian connections:
\begin{enumerate}
\item $t=0$, it is known as the \emph{first canonical connection of Lichnerowicz}.
\item $t=1$, it is known as the \emph{second canonical connection of Lichnerowicz}. When $J$ is integrable, it is nothing but the \emph{Chern connection}.
\item $t=-1$, this is the \emph{Strominger-Bismut connection}.
\item $t=1/2$, it has been called the \emph{conformal connection} by Libermann.
\item $t=1/3$, this is the Hermitian connection that minimizes the norm of its torsion tensor.
\end{enumerate}
When $X$ is K\"ahler, this line collapses to a single point, i.e. the Levi-Civita connection.
\end{thm}

In our case, $J$ is always integrable thus $(\ud^c\omega)^+=\ud^c\omega$. Therefore we have the following simplified expression \[\nabla^t=\nabla^1+\frac{t-1}{4}(\ud^c\omega+\frM(\ud^c\omega)).\]

\section{Flat Invariant Solutions}

In this section, we will solve the Strominger system on complex Lie groups using the ansatz proposed in \cite{biswas2013}. As we will see, this is the most natural and symmetric solution one can expect. The name ``flat'' comes from the assumption that the extra bundle $E$ is flat, i.e., $F\equiv0$. Under such assumption, the Hermitian-Yang-Mills equation (\ref{hym}) is satisfied automatically, and therefore the Strominger system reduces to the following equations
\begin{eqnarray}
\label{rac}\sqrt{-1}\pt\bpt\omega=\frac{\alpha'}{4}\tr~R\wedge R,\\
\ud(\|\Omega\|_\omega\cdot\omega^{n-1})=0.
\end{eqnarray}

Now we assume that $X$ is a complex Lie group and let $e\in X$ be the neutral element. Obviously $X$ is holomorphically parallelizable, hence it has trivial canonical bundle. Given any Hermitian metric on $T_eX$, we can translate it to get a left-invariant Hermitian metric on $X$. Let us still denote the associated Hermitian form by $\omega$. It follows that under such metric, $\|\Omega\|_\omega$ is a constant and the conformal balanced equation (\ref{cb}) dictates that $\omega$ is a balanced metric. The straightforward calculation from \cite{abbena1986} shows that $\omega$ is balanced if and only if $X$ is unimodular. Moreover this condition is independent of the choice of the left-invariant metric.

From now on we will assume that $X$ is a unimodular complex Lie group and $\omega$ is left-invariant. So Equation (\ref{cb}) holds and we only have to consider the reduced anomaly cancellation equation (\ref{rac}). The new idea here is to use the canonical 1-parameter family of Hermitian connections described in Section 2 to compute $R$. In order to do that let us fix some notations first.

Let $\frg$ be the complex Lie algebra associated with $X$ and let $e_1,\dots,e_n\in\frg$ be an orthonormal basis under the given left-invariant metric. In addition we define the structure constants $c^k_{ij}$ in the usual way \[[e_i,e_j]=c^k_{ij}e_k.\] Let $\{e^i\}_{i=1}^n$ be the holomorphic 1-forms on $X$ such that $e^i(e_j)=\sqrt{2}\delta^i_j$. Then we can express the Hermitian form $\omega$ as \[\omega=\frac{\sqrt{-1}}{2}\sum_{i=1}^ne^i\wedge\be^i.\] Furthermore, the Maurer-Cartan equations give \begin{equation}\label{maurer}
\ud e^i=-\frac{1}{\sqrt{2}}\sum_{j<k}c^i_{jk}e^j\wedge e^k.
\end{equation}

Now we shall compute the canonical 1-parameter family of Hermitian connections $\nabla^t$. We may trivialize the holomorphic tangent bundle $T_\ccc X$ by $\{e_i\}_{i=1}^n$. Under such trivialization, the Chern connection $\nabla^1$ is simply $\ud$ and we thus get \[\nabla^t=\ud+\frac{t-1}{4}(\ud^c\omega+\frM(\ud^c\omega))\triangleq\ud+A^t.\]

Now \[\begin{split}\ud^c\omega&=\sqrt{-1}(\bpt-\pt)\omega=\frac{1}{2}\sum_i\ud e^i\wedge\be^i+e^i\wedge\ud\be^i\\ &=-\frac{1}{2\sqrt{2}}\sum_i\sum_{j<k}(c^i_{jk}e^j\wedge e^k\wedge\be^i+\overline{c^i_{jk}}e^i\wedge\be^j\wedge\be^k)\\ &=-\frac{1}{2\sqrt{2}}\sum_i\sum_{j<k}c^i_{jk}(e^j\otimes(e^k\wedge\be^i)-e^k\otimes(e^j\wedge\be^i)+\be^i\otimes (e^j\wedge e^k))+\mathrm{conjugate},\end{split}\] and therefore \[\ud^c\omega+\frM(\ud^c\omega)=-\frac{1}{\sqrt{2}}\sum_{i,j,k}c^i_{jk}e^j\otimes(e^k\wedge\be^i)+\mathrm{conjugate}.\]

If we write $e^i=x^i-\sqrt{-1}Jx^i$, then $\{x^i,Jx^i\}_{i=1}^n$ form a real orthonormal frame of $T_\rr^*X$, and \[\begin{split}\ud^c\omega+\frM(\ud^c\omega)=-\sqrt{2}\sum_{i,j,k}&\re(c^i_{jk})\left(x^j\otimes(x^k\wedge x^i+Jx^k\wedge Jx^i)+Jx^j\otimes(x^k\wedge Jx^i-Jx^k\wedge x^i)\right)\\ -&\ima(c^i_{jk})\left(x^j\otimes(x^k\wedge Jx^i-Jx^k\wedge x^i)-Jx^j\otimes(x^k\wedge x^i+Jx^k\wedge Jx^i)\right).\end{split}\]

Using \[A^t(U,V,W)=\langle A^t(U)V,W\rangle\] that identifies $A^t\in\Omega^2(T_\rr X)$ as an element in $\Omega^1(\endo~T_\ccc X)$, we can rewrite the above equality as \[\begin{split}A^t&=\frac{1-t}{2\sqrt{2}}\sum_{i,j,k}\re(c^i_{jk})(x^j\otimes A_{ki}+\sqrt{-1}Jx^j\otimes S_{ki})-\ima(c^i_{jk})(\sqrt{-1}x^j\otimes S_{ki}-Jx^j\otimes A_{ki})\\ &=\frac{1-t}{4\sqrt{2}}\sum_{i,j,k}e^j\otimes c^i_{jk}(A_{ki}-S_{ki})+\be^j\otimes\overline{c^i_{jk}}(A_{ki}+S_{ki})\\ &=\frac{t-1}{2\sqrt{2}}\sum_{i,j,k}c^i_{jk}e^j\otimes E_{ki}-\overline{c^i_{jk}}\be^j\otimes E_{ik}.\end{split}\] Here, $A_{ki}$ is the skew-symmetric matrix whose $(i,k)$-entry is 1 and $(k,i)$-entry is -1, $E_{ki}$ the matrix whose $(k,i)$-entry is 1 and $S_{ki}$ the symmetric matrix whose both $(i,k)$ and $(k,i)$-entries are 1. If $k=i$, $S_{kk}$ is the matrix with $(k,k)$-entry being 2. All other entries not mentioned above vanish.

It is straightforward to verify that the above expression gives exactly \begin{equation}\label{exp}A^t=\frac{t-1}{2\sqrt{2}}\sum_ie^i\otimes \ad(e_i)^T-\be^i\otimes\overline{\ad(e_i)}.\end{equation} Consequently, \[R^t=\ud A^t+A^t\wedge A^t=\frac{t-1}{2\sqrt{2}}\sum_i\ud e^i\otimes \ad(e_i)^T-\ud\be^i\otimes\overline{\ad(e_i)}+A^t\wedge A^t.\] As $\tr~A^t\wedge A^t=0$, it follows directly from unimodularity that the first Chern form \[c_1=\frac{\sqrt{-1}}{2\pi}\tr~R^t=0.\]

\begin{rmk}
From the expression of $A^t$, we know that, as an element of $\Omega^1(\endo~T_\ccc X)$, $A^t$ does not depend on the left-invariant metric we begin with. It follows that $R^t=\ud A^t+A^t\wedge A^t$ does not depend on the metric either. However the canonical 1-parameter family of Hermitian connections does depend on the choice of the metric.
\end{rmk}

Now we want to compute \[\tr~R^t\wedge R^t=\tr~\ud A^t\wedge\ud A^t+2\cdot\tr~A^t\wedge A^t\wedge\ud A^t+\tr~A^t\wedge A^t\wedge A^t\wedge A^t.\] It is well-known that the last term $\tr~A^t\wedge A^t\wedge A^t\wedge A^t$ is 0. Let us compute the first two terms separately.

The first term is \[\begin{split}\tr~\ud A^t\wedge\ud A^t&=\frac{(t-1)^2}{8}\sum_{i,j}\tr\left((\ud e^i\otimes\ad(e_i)^T-\ud\be^i\otimes\overline{\ad(e_i)})\wedge(\ud e^j\otimes\ad(e_j)^T-\ud\be^j\otimes\overline{\ad(e_j)})\right)\\ &=\frac{(t-1)^2}{8}\sum_{i,j}\ud e^i\wedge\ud e^j\cdot\tr\left(\ad(e_i)^T\ad(e_j)^T\right)-\ud e^i\wedge\ud\be^j\cdot\tr\left(\ad(e_i)^T\overline{\ad(e_j)}\right)\\ &+\textrm{conjugate of the above line}.
\end{split}\]

\begin{prop}\label{prop1}
\[\sum_{i,j}\ud e^i\wedge\ud e^j\cdot\tr(\ad(e_i)^T\ad(e_j)^T)=0.\]
\end{prop}
\begin{proof}
We first make two observations. For the dimension in which physicists are most interested, i.e. $n=3$, Proposition \ref{prop1} is trivially true since $\ud e^i\wedge\ud e^j=0$. If $X$ is nilpotent (hence unimodular), the above equation holds because $\tr(\ad(e_i)^T\ad(e_j)^T)=\kappa(e_i,e_j)=0$, where $\kappa$ is the Killing form.

For the general case, we use Equation (\ref{maurer}) to expand the LHS and we only have to proof the following identity \[\sum_{i,j,r,s,a,b,c,d}c^i_{ab}c^j_{cd}c^r_{is}c^s_{jr}\cdot e^a\wedge e^b\wedge e^c\wedge e^d\triangleq\sum_{a,b,c,d}F_{abcd}\cdot e^a\wedge e^b\wedge e^c\wedge e^d=0.\] Like Riemannian curvature tensor, $F_{abcd}$ has many symmetries. It is straightforward from the definition that \[F_{abcd}=-F_{bacd}=-F_{abdc}=F_{cdab}.\] It follows that Proposition \ref{prop1} is equivalent to the Bianchi identity \[F_{abcd}+F_{acdb}+F_{adbc}=0.\] Using the Jacobi identity \[c^i_{jk}c^r_{il}+c^i_{kl}c^r_{ij}+c^i_{lj}c^r_{ik}=0\] repetitively, we deduce that\footnote{We adopt the Einstein notation for summation here.} \[\begin{split}F_{abcd}&=(c^i_{ab}c^r_{is})c^j_{cd}c^s_{jr}=-(c^i_{bs}c^r_{ia}+c^i_{sa}c^r_{ib})c^j_{cd}c^s_{jr}= c^j_{cd}(c^s_{jr}c^i_{sb})c^r_{ia}-c^j_{cd}(c^s_{jr}c^i_{sa})c^r_{ib}\\ &=-c^j_{cd}(c^s_{rb}c^i_{sj}+c^s_{bj}c^i_{sr})c^r_{ia}+c^j_{cd}(c^s_{ra}c^i_{sj}+c^s_{aj}c^i_{sr})c^r_{ib}\\ &=c^j_{cd}c^i_{sj}(c^r_{ib}c^s_{ra}+c^r_{ai}c^s_{rb})+c^j_{cd}c^i_{sr}(c^s_{jb}c^r_{ia}-c^s_{ja}c^r_{ib})\\ &=c^j_{cd}c^i_{sj}c^r_{ab}c^s_{ri}+c^j_{cd}c^i_{sr}(c^s_{jb}c^r_{ia}-c^s_{ja}c^r_{ib})\\ &=-F_{abcd}+c^j_{cd}c^i_{sr}(c^s_{jb}c^r_{ia}-c^s_{ja}c^r_{ib}).\end{split}\]
Using the symmetry $F_{abcd}=F_{cdab}$, we get \[4F_{abcd}=c^j_{cd}c^i_{sr}(c^s_{jb}c^r_{ia}-c^s_{ja}c^r_{ib})+c^j_{ab}c^i_{sr}(c^s_{jd}c^r_{ic}-c^s_{jc}c^r_{id}).\]
Rotate the indices $(b,c,d)$ and sum them up, after a rearrangement of terms, we get \[\begin{split}4(F_{abcd}+F_{acdb}+F_{adbc})&=(c^j_{cd}c^s_{jb}+c^j_{db}c^s_{jc}+c^j_{bc}c^s_{jd})c^r_{ia}c^i_{sr} -(c^j_{cd}c^s_{ja}+c^j_{da}c^s_{jc}+c^j_{ac}c^s_{jd})c^r_{ib}c^i_{sr} \\ &+(c^j_{ab}c^s_{jd}+c^j_{bd}c^s_{ja}+c^j_{da}c^s_{jb})c^r_{ic}c^i_{sr} -(c^j_{ab}c^s_{jc}+c^j_{bc}c^s_{ja}+c^j_{ca}c^s_{jb})c^r_{id}c^i_{sr}\\ &=0.\end{split}\]
\end{proof}
As a consequence of Proposition \ref{prop1}, we conclude \[\tr~\ud A^t\wedge\ud A^t=-\frac{(t-1)^2}{4}\sum_{i,j}\ud e^i\wedge\ud\be^j\cdot\tr(\ad(e_i)^T\overline{\ad(e_j)}).\]

Now we proceed to compute the second term. \[\begin{split}2\cdot\tr~A^t\wedge A^t\wedge\ud A^t&=\frac{(t-1)^3}{16\sqrt{2}}\sum_{i,j,k}\tr\left\{\left(\ud e^i\otimes\ad^T(e_i)-\ud\be^i\otimes\overline{\ad(e_i)}\right)\wedge\right.\\ &~\left.\left(e^j\wedge e^k\otimes\ad[e_k,e_j]^T+\be^j\wedge\be^k\otimes\overline{\ad[e_j,e_k]}-2e^j\wedge\be^k\otimes[\ad(e_j)^T,\overline{\ad( e_k})]\right)\right\}.\end{split}\]

Like before, we have the following
\begin{prop}\label{prop2}
\[\sum_{i,j,k}\ud e^i\wedge e^j\wedge e^k\cdot\tr(\ad(e_i)^T\ad[e_k,e_j]^T)=0.\]
\end{prop}
\begin{proof}
This is actually equivalent to Proposition \ref{prop1}.
\end{proof}
\begin{prop}\label{prop3}
\[\sum_{i,j}\ud e^i\wedge e^j\wedge\be^k\cdot\tr(\ad(e_i)^T[\ad(e_j)^T,\overline{\ad(e_k)}])=0.\]
\end{prop}
\begin{proof}
The proof is very similar to the one of Proposition \ref{prop1} but with less complexity.
\end{proof}
\begin{prop}\label{prop4}
\[\sum_{j,k}\ud e^i\wedge\be^j\wedge\be^k\cdot\tr(\ad(e_i)^T\overline{\ad[e_j,e_k]})=-2\sqrt{2}\sum_l\ud e^i\wedge\ud\be^l\cdot\tr(\ad(e_i)^T\overline{\ad(e_l)}).\]
\end{prop}
\begin{proof}
Direct calculation.
\end{proof}
Combining Propositions \ref{prop2}, \ref{prop3} and \ref{prop4}, we get \[\begin{split}2\cdot\tr~A^t\wedge A^t\wedge\ud A^t&=-\frac{(t-1)^3}{4}\sum_{i,j}\ud e^i\wedge\ud\be^j\cdot\tr(\ad(e_i)^T\overline{\ad(e_j)}),\\ \tr~R^t\wedge R^t&=-\frac{t(t-1)^2}{4}\sum_{i,j}\ud e^i\wedge\ud\be^j\cdot\tr(\ad(e_i)^T\overline{\ad(e_j)}).\end{split}\]
\begin{cor}
When we choose the Hermitian connection to be either the Chern connection ($t=1$) or the first canonical Lichnerowicz connection ($t=0$), we have $\tr~R\wedge R=0$ and thus the Strominger has no solution using our ansatz. This generalizes the result in \cite{andreas2014}.
\end{cor}
\begin{rmk}
From calculation above, it is tempting to conjecture that $\tr(R^t)^k$ is always a real $(k,k)$-form. However, $R^t$ itself in general contains both $(2,0)$ and $(0,2)$ parts, and therefore it does not satisfy the so-called equation of motion derived from the heterotic string effective action.
\end{rmk}

As \[\sqrt{-1}\pt\bpt\omega=\frac{1}{2}\sum_i\ud e^i\wedge\ud\be^i,\] the anomaly cancellation equation (\ref{ac}) reduces to \begin{equation}\label{rrac}\sum_i\ud e^i\wedge\ud\be^i=-\frac{t(t-1)^2}{8}\alpha'\sum_{j,k}\ud e^j\wedge\ud\be^k\cdot\tr(\ad(e_j)^T\overline{\ad(e_k)}).\end{equation}

It seems that in general whether Equation (\ref{rrac}) has a solution or not is not easy to answer. However, we have the following result.
\begin{thm}\label{thm1}
If we further assume that $X$ is semisimple, then there is a unique left-invariant Hermitian metric up to scaling, i.e., the one coming from the Killing form, such that our ansatz of solution does exist. If we pick $t<0$, for instance the Strominger-Bismut connection, we obtain solutions with $\alpha'>0$; if we pick $t>0$ with $t\neq1$, we get solutions with $\alpha'$ negative.
\end{thm}
\begin{proof}
When $X$ is semisimple, then $\{\ud e^1,\dots,\ud e^n\}$ are linear independent 2-forms. Therefore Equation (\ref{rrac}) requires that $\tr(\ad(e_i)^T\overline{\ad(e_j)})=c\delta_{ij}$ for some positive $c$. This determines the metric uniquely.
\end{proof}

We can say a little more about Equation (\ref{rrac}) in complex dimension 3. Actually we can classify all the 3-dimensional complex unimodular Lie algebras.

\begin{prop}
Let $\frg$ be a 3-dimensional unimodular Lie algebra over $\ccc$, then $\frg$ must be isomorphic to one of the follows:
\begin{enumerate}
\item $\frg$ is abelian,
\item $\dim_\ccc[\frg,\frg]=1$, $\frg=\spann\{h,x,y\}$ with $[h,x]=[h,y]=0,~[x,y]=h$,
\item $\dim_\ccc[\frg,\frg]=2$, $\frg=\spann\{h,x,y\}$ with $[h,x]=x,~[h,y]=-y,~[x,y]=0$,
\item $\frg=\spl_2\ccc$.
\end{enumerate}
\end{prop}

\begin{rmk}
Cases (a), (b), (c) and (d) each corresponds to the abelian, nilpotent, solvable and semisimple Lie algebra respectively. They are listed in Page 28 of \cite{knapp2002}. For their Lie groups, Case (a) and (d) are well-known. Case (b) corresponds to the Heisenberg group and Case (c) corresponds to the complexification of the group of rigid motions on $\rr^2$.
\end{rmk}

For case (a), any invariant metric is actually K\"ahler and our ansatz solves the Strominger system because both sides of the anomaly cancellation equation (\ref{ac}) are 0. Case (d) has been treated in Theorem \ref{thm1}, so we only discuss the other two situations.

For Case (b), we have $Z(\frg)=[\frg,\frg]$ is 1-dimensional, and we may assume it is spanned by $e_1$. Under such assumption, the only nontrivial structure constant is $c^1_{23}=-c^1_{32}\neq0$, others are 0. It follows that $\ad(e_1)=0$ and $\ud e^2=\ud e^3=0$. One can calculate easily that $\tr~R^t\wedge R^t=0$ while $R^t\neq0$ for $t\neq1$. This gives an example of a non-flat connection on a bundle such that all the Chern forms are 0. In particular Equation (\ref{rrac}) has no solution in this case.

For Case (c), $[\frg,\frg]$ is 2-dimensional. Without loss of generality, we may assume it is spanned by $\{e_1,e_2\}$. Under such assumption, we have \[\ad(e_1)=\begin{pmatrix}0&&\\ &0&\\
\alpha&\beta&0\end{pmatrix},\quad \ad(e_2)=\begin{pmatrix}0&&\\ &0&\\ \gamma&\!-\!\alpha&0\end{pmatrix},\quad \ad(e_3)=-\begin{pmatrix}\alpha&\beta&\\ \gamma&\!-\!\alpha&\\ &&0 \end{pmatrix}\] with $\alpha^2+\beta\gamma\neq0$. In addition, we have the following formulae for computing exterior derivatives: \[\begin{split}\ud e^1&=-\frac{1}{\sqrt{2}}(\alpha e^1\wedge e^3+\gamma e^2\wedge e^3),\\ \ud e^2&=-\frac{1}{\sqrt{2}}(\beta e^1\wedge e^3-\alpha e^2\wedge e^3),\\ \ud e^3&=0.\end{split}\] It follows that $\ud e^1$ and $\ud e^2$ are linearly independent and \[\begin{split}\sqrt{-1}\pt\bpt\omega&=\frac{1}{2}(\ud e^1\wedge\ud\be^1+\ud e^2\wedge\ud\be^2),\\ \tr~R^t\wedge R^t&=-\frac{t(t-1)^2}{4}\sum_{i,j=1,2}\ud e^i\wedge\ud\be^j\cdot\tr(\ad(e_i)\overline{\ad(e_j)}^T).\end{split}\]

It follows that the anomaly cancellation equation has a solution if and only if $\ad(e_1)$ and $\ad(e_2)$ are orthonormal (up to a positive scalar) under the metric $\langle x,y\rangle=\tr(x\bar{y}^T)$. Or equivalently, under the induced metric, $\ad(h):[\frg,\frg]\to[\frg,\frg]$ is unitary (up to a positive scalar)\footnote{Note this condition does not depend on the choice of $h$}.

To summarize, we have the following result:
\begin{thm}
For any Lie group with Lie algebra (b), there is no flat invariant solution to the Strominger system.
For any Lie group with Lie algebra (c) with the basis $\{h,x,y\}$ chosen. As long as $x$ and $y$ are orthogonal to each other in the Hermitian metric, our ansatz solves the Strominger system with $\alpha'>0$ for $t<0$ and $\alpha'<0$ for $t>0$ and $t\neq1$.
\end{thm}
\begin{rmk}
As our ansatz are invariant under left translation, solutions to the Strominger system on $X$ descend to solutions on the quotient $\Gamma\backslash X$ for any discrete closed subgroup $\Gamma$. By Wang's classification theorem \cite{wang1954b}, such quotients include all the compact complex parallelizable manifolds.
\end{rmk}

\section{Non-flat Invariant Solutions}

In this section, we will consider invariant solutions to the Strominger system with nontrivial $F$.

Let $\rho:X\to GL_n\ccc$ be a faithful holomorphic representation, then $X$ naturally acts on $\ccc^n$ from right by setting $v\cdot g:=\overline{\rho(g)}^Tv$ for $g\in X$ which we abbreviate to $\bar{g}^Tv$. Consider the following Hermitian metric $H$ defined on the trivial bundle $E=X\times\ccc^n$: at a point $g\in X$, the metric is given by \[\langle v,w\rangle_g=(v\cdot g)^T\bar{B}\overline{(w\cdot g)}=v^T\bar{g}\bar{B}g^T\bar{w},\] where $B=\bar{B}^T$ is some fixed positive Hermitian matrix, $v,w\in\ccc^n$ are arbitrary column vectors. Choose the standard basis for $\ccc^n$ as a holomorphic trivialization, then \[H_g=(h_{i\bar\jmath})_g=\bar{g}\bar{B}g^T.\] Let us compute its curvature $F$ with respect to the Chern connection. Now $F^{0,2}=F^{2,0}=0$ is satisfied automatically, and by the formula $F=\bpt(\bar{H}^{-1}\pt\bar{H})$, we get \[\begin{split}F&=\bpt[(\bar{g}^T)^{-1}(B^{-1}g^{-1}\pt g\cdot B)\bar{g}^T]\\ &=-(\bar{g}^T)^{-1}[(\bpt\bar{g}^T\cdot(\bar{g}^T)^{-1})(B^{-1}g^{-1}\pt g\cdot B)+(B^{-1}g^{-1}\pt g\cdot B)(\bpt\bar{g}^T\cdot(\bar{g}^T)^{-1})]\bar{g}^T.\end{split}\]

Notice that $g^{-1}\pt g$ is the Maurer-Cartan form \[g^{-1}\pt g=\frac{1}{\sqrt{2}}\sum_ie^i\otimes e_i.\] Therefore \[F=-\frac{1}{2}(\bar{g}^T)^{-1}\left(\sum_{i,j}e^i\wedge\be^j\otimes[B^{-1}e_iB,\be_j^T]\right)\bar{g}^T\] and thus $\tr~F=0$. Furthermore, if we set $e'_m=B^{-1/2}e_mB^{1/2}$, then we have \[\begin{split} \tr~F\wedge F&=\frac{1}{4}\sum_{i,j,k,l}e^i\wedge\be^k\wedge e^j\wedge\be^l\cdot\tr([B^{-1}e_iB,\be_k^T][B^{-1}e_jB,\be_l^T])\\ &=-\frac{1}{8}\sum_{i,j,k,l}e^i\wedge e^j\wedge\be^k\wedge\be^l\cdot\tr([e'_i,\overline{e'_k}^T][e'_j,\overline{e'_l}^T] -[e'_i,\overline{e'_l}^T][e_j',\overline{e'_k}^T])\\ &=-\frac{1}{8}\sum_{i,j,k,l}e^i\wedge e^j\wedge\be^k\wedge\be^l\cdot\tr([e'_i,e'_j][\overline{e'_k}^T,\overline{e'_l}^T])\\ &=\frac{1}{8}\sum_{i,j,k,l}c^m_{ij}\overline{c^n_{kl}}e^i\wedge e^j\wedge\be^k\wedge\be^l\cdot\tr(e'_m\overline{e'_n}^T)\\ &=\sum_{m,n}\ud e^m\wedge\ud\be^n\cdot\tr(e'_m\overline{e'_n}^T).\end{split}\] Similar calculation shows that the Hermitian-Yang-Mills equation (\ref{hym}) is equivalent to \[\sum_i[e'_i,\overline{e'_i}^T]=0.\]

From now on, we will restrict ourselves to the $SL_2\ccc$ case.

Let $\rho$ be the standard representation on $\ccc^2$. The Strominger equation is reduced to
\begin{eqnarray}
\label{rhym}&\sum_i[e'_i,\overline{e'_i}^T]=0,\\ \label{rrrac} &-\dfrac{2}{\alpha'}\delta_{ij}=\dfrac{t(t-1)^2}{4}\tr(\ad(e_i)\overline{\ad(e_j)}^T)+\tr(e'_i\overline{e'_j}^T).
\end{eqnarray}

If we set the metric on $\spl_2\ccc$ to be $\langle U,V\rangle=\tr~U\overline{V}^T$ and $B=\id$, then (\ref{rhym}) holds and the three terms in (\ref{rrrac}) are proportional. As along as $t(t-1)^2+4\neq0$, we may choose the coupling constant $\alpha'$ properly to obtain invariant non-flat solutions to the Strominger system. If $t=0\textrm{ or }1$, then (\ref{rrrac}) becomes \[-\dfrac{2}{\alpha'}\delta_{ij}=\tr(e'_i\overline{e'_j}^T)\] and more solutions can be found. In fact, if we identify $\spl_2\ccc\cong\ccc^3\cong\mathrm{Sym}^2\ccc^2$, then as long as the metric on $\spl_2\ccc$ is induced from a metric on $\ccc^2$, our ansatz gives a solution to the Strominger system with vanishing $R$ and negative $\alpha'$.

\begin{rmk}
It is well-known that all the irreducible representations of $SL_2\ccc$ are generated by the standard representation on $\ccc^2$. Therefore from any solution above, we can produce solutions to the Strominger system with trivial bundle $E$ of arbitrary rank. In addition, this argument generalizes to many other complex semisimple Lie groups.
\end{rmk}

\begin{rmk}
Some new phenomenon occur when $\frg$ is the Heisenberg algebra (b). For example, $X=X^H$ is the Heisenberg group \[X^H=\left\{\begin{pmatrix}1&a&b\\ &1&c\\ &&1\end{pmatrix}:a,b,c\in\ccc\right\}.\] It can be checked directly that the anomaly cancellation equation (\ref{ac}) can always be solved by choosing $\alpha'<0$ properly. However the Hermitian-Yang-Mills equation (\ref{rhym}) is hard to solve. For example, let $\rho_a:X^H\to GL_3\ccc$ be the representation given by ``conjugation by $a$'', i.e., $\rho_a(g)=aga^{-1}$ for some $a\in GL_3\ccc$, $g\in X^H$. We would like to find $a$ such that Equation (\ref{rhym}) has a solution. If we think of entries of $a$ as unknowns, then Equation (\ref{rhym}) can be rewritten as a system of degree 6 real polynomial equations, which is not easy to solve. In a very simple case that $B=a=\id$, we can never find $e_1,e_2,e_3$ such that the Hermitian-Yang-Mills equation holds.
\end{rmk}

\begin{rmk}
Let $\Gamma$ be a discrete subgroup of $X$, then $E'=\ccc^n\times_\Gamma X$ can be naturally viewed as a vector bundle over $X'=\Gamma\backslash X$. Moreover, one can see from our construction that the metric $H$ descends to an Hermitian metric on the vector bundle $E'$. Unfortunately, it seems that there is no natural holomorphic structure on the total space of $E'$ and we can not na\"ively obtain solutions of the Strominger system on $X'$ in this way. If we modify the right action of $X$ on $\ccc^n$ by $v\cdot g=\rho(g)^Tv$, then the holomorphic structure does descend on $E'$ and we get a holomorphic vector bundle $E'$ over $X'$. However, the price to pay is that the similarly constructed metric on $E$ turns out to be flat and it reduces to the situation of Section 3.
\end{rmk}

\bibliographystyle{alpha}

\bibliography{C:/Users/Piojo/Dropbox/Documents/Source}
~\\
\begin{tabular}{lcl}
\\
Teng Fei &\quad& Shing-Tung Yau \\
Department of Mathematics &\quad& Department of Mathematics \\
Massachusetts Institute of Technology &\quad\quad\quad\quad\quad& Harvard University \\
Cambridge, MA 02139, USA &\quad& Cambridge, MA 02138, USA \\
{\tt{tfei@math.mit.edu}} &\quad& {\tt{yau@math.harvard.edu}}
\end{tabular}

\end{document}